\documentclass[aps,prl,twocolumn,superscriptaddress,floatfix,nofootinbib,showpacs,longbibliography,10pt]{revtex4-2}

\usepackage[utf8]{inputenc}
\usepackage[T1]{fontenc}     
\usepackage[british]{babel}  

\usepackage[sc,osf]{mathpazo}
\usepackage{times}

\usepackage[table,dvipsnames]{xcolor}

\usepackage[scaled=0.86]{berasans}

\usepackage{makecell}
\usepackage{comment}
\usepackage{enumitem}
\usepackage{graphicx}
\usepackage[babel]{microtype}

\usepackage{amsmath,amssymb,amsthm,bm,amsfonts,mathrsfs,bbm}

\usepackage{xspace}
\usepackage{pgf,tikz}
\usepackage{multirow}
\usepackage{array}
\usepackage{bigstrut}
\usepackage{braket}
\usepackage{mathtools}
\usepackage[caption=false]{subfig}
\usepackage{natbib}
\usepackage{centernot}

\usepackage[colorlinks=true,citecolor=Magenta,urlcolor=Blue]{hyperref}

\makeatletter
\newcommand{\setword}[2]{%
  \phantomsection
  #1\def\@currentlabel{\unexpanded{#1}}\label{#2}%
}
\makeatother

\newcommand{\be}{\begin{equation}}
\newcommand{\ee}{\end{equation}}
\newcommand{\ba}{\begin{eqnarray}}
\newcommand{\ea}{\end{eqnarray}}

\newtheorem{theorem}{Theorem}
\newtheorem{corollary}{Corollary}

\newtheorem{proposition}{Proposition}
\newtheorem{observation}{Observation}

\newtheorem{lemma}{Lemma}

\def\>{\rangle}
\def\<{\langle}

\begin{document}

\title{Exclusion reshapes the operational manifestation of preparation contextuality}

\author{Pritam Roy}
\email{roy.pritamphy@gmail.com}
\affiliation{S. N. Bose National Centre for Basic Sciences, Block JD, Sector III, Salt Lake, Kolkata 700 106, India}

\author{Thansingh Jankawat}
\email{thansinghjankawat@gmail.com}
\affiliation{S. N. Bose National Centre for Basic Sciences, Block JD, Sector III, Salt Lake, Kolkata 700 106, India}

\author{Ranendu Adhikary}
\email{ronjumath@gmail.com}
\affiliation{Department of Physics and Center for Quantum Frontiers of Research and Technology (QFort),
National Cheng Kung University, Tainan 701, Taiwan}

\author{A. S. Majumdar}
\email{archan@bose.res.in}
\affiliation{S. N. Bose National Centre for Basic Sciences, Block JD, Sector III, Salt Lake, Kolkata 700 106, India}

\begin{abstract}
Replacing the task of retrieval with exclusion changes how preparation contextuality manifests operationally under parity-oblivious constraints, with exclusion showing a quantum advantage where retrieval does not. We introduce the parity-oblivious random exclusion code (POREC) and show that for prime symbol size $m$, classical and preparation-noncontextual encodings provide a tight noncontextual bound. For the first nontrivial case (two digits, three symbols) our derived exact qubit optimum  violates this bound, in contrast  to parity-oblivious retrieval which displays no quantum advantage. This characteristic difference is absent without parity constraints.  For general prime $m$, qubit strategies achieve a quantum-to-noncontextual gap that grows linearly relative to the random exclusion code (REC) gap, exceeding both parity-oblivious retrieval and standard REC. The exact qubit bound  yields a sharp semi-device-independent certification of dimension $d \geq 3$. Our analysis of  noise robustness demonstrates POREC 
to be amenable for experimental implementation on existing prepare-and-measure platforms, establishing parity-oblivious exclusion as  a distinct operational probe of preparation 
contextuality, as well as a practical information processing protocol
with wide applications.
\end{abstract}
\maketitle

{\it Introduction. --} Prepare-and-measure (PM) communication tasks~\cite{Holevo73,NielsenChuang2010,Watrous2018} provide an operational framework to probe and certify fundamental nonclassical features of quantum systems, such as contextuality~\cite{BELL1966,KOCHEN1968,Spekkens2005Contextuality, Pusey2018,SchmidEquv2018,Hazra2026efficient}, violation of  macrorealism~\cite{Leggett_macrorealism1985,Emary_2014} and measurement
incompatibility~\cite{Saha2023}. Moreover, the PM scenario enables quantum information processing  applications without employing the expensive and
operationally fragile resource of quantum entanglement, such as  random access codes~\cite{ambainis2009quantumrandomaccesscodes,Tavakoli2015QRAC, Bera2022},  semi-device-independent   self-testing~\cite{Farkas2019Selftesting, Tavakoli2020NonProj},
randomness generation~\cite{Li2011sDIrandom,Li2012sDIrandom,PPNath2024},
quantum key distribution~\cite{Pawlowski2011SDI}, communication complexity~\cite{Sahacommunicationadvantage2019,Saha_2019,Gupta2023}, and dimension witnessing~\cite{Brunner2013DW,Bowles2014NonlinearDW, deVicente2019DimensionBound,Tavakoli2021PM}.

Within the PM framework, parity-oblivious communication tasks~\cite{Spekkens2009POM,Ambainis2019dPORAC,Hameedi2017dPORAC,Ghorai2018POMOptimal,Roy2024GPOC,Singh_2025,patra2025contextuality} play a central role, constraining the sender to encode information about the input while revealing no multi-digit parities. Violations of their optimal classical bound certify preparation contextuality~\cite{Spekkens2005Contextuality,Pusey2018,SchmidEquv2018,Hazra2026efficient}, a form of generalized noncontextuality rooted in the \emph{Leibniz’s Principle of the Identity of Indiscernibles}~\cite{leibniz1991discourse,spekkens2019indiscernibles}, which requires that operationally indistinguishable preparations be represented identically at the ontological level. Such violations provide a direct operational signature of nonclassicality. Most existing works herein focus on \emph{retrieval} tasks, where the receiver identifies one of the encoded symbols, and is therefore, strongly constrained by parity-obliviousness.

A distinct operational paradigm is provided by \emph{exclusion} tasks, where the receiver outputs a value different from the queried symbol. Unlike retrieval, exclusion is less restrictive, can exhibit quantum advantages in success probability~\cite{Emeriau2022,Bae2026}, and connects directly to contextuality. Conclusive exclusion scenarios inspired by the Pusey-Barrett-Rudolph construction~\cite{Pusey2012RealityQuantumState} typically assuming preparation independence, violate generalized noncontextuality inequalities~\cite{contextualadvantageconclusiveexclusion2026}. Since exclusion reduces to retrieval for binary symbols, the minimal nontrivial setting is the two-digit, $m \ge 3$ regime, relevant for high-rate randomness generation~\cite{Random_number_generator_review2017,Brown2021} and multi-outcome quantum key distribution~\cite{Ding2017HDQKD,Tan2021SecureKeyRates,Brown2021}. 

Random Exclusion Codes (REC)~\cite{Bae2026}, as a counterpart of Random Access Codes (RAC),  differ qualitatively from the latter, though both exhibit the same quantum-classical gap in success probability in the unconstrained setting~\cite{Bae2026}. REC can display a communication dimension advantage under perfect and uniform exclusion, absent in RAC~\cite{Saha2023}, indicating a distinct operational structure. This raises the question of whether such equivalence of quantum-classical gaps persists under parity-oblivious constraints, which forbid any leakage of multi-digit parity information. While such constraints severely limit retrieval, exclusion only requires ruling out the correct value and may still exploit residual information. It is therefore natural to ask: \emph{does replacing retrieval with exclusion alter characteristically how preparation contextuality manifests under parity-oblivious constraints?}

In this Letter, we answer this question in the affirmative by introducing the \emph{parity-oblivious random exclusion code} (POREC), combining the symmetry constraints of parity-oblivious  tasks~\cite{Spekkens2009POM,Chailloux2016PORAC,Hameedi2017dPORAC,
Ambainis2019dPORAC} with the exclusion objective of random exclusion coding~\cite{Bae2026}.
We perform an exact analytic treatment focusing on prime $m$, where $\mathbb{Z}_m$ forms a field and ensures uniform parity classes. We derive a tight classical (equivalently, preparation-noncontextual) bound under full parity-obliviousness for $n$-digit inputs over a prime symbol $m \ge 3$,  showing that all admissible encodings reduce to a single-digit form. In the first nontrivial two-digit, three-symbol case, the exact qubit optimum violates this bound, witnessing preparation contextuality, while retrieval under identical constraints and dimension shows no violation~\cite{Hameedi2017dPORAC}. Structurally, the same additive encoding requires alignment with one component for retrieval, which cannot be maintained across inputs, while exclusion succeeds via anti-aligned qubit measurements that only need to rule out outcomes.

Although REC and RAC share the same quantum-classical gap in the unconstrained setting~\cite{Bae2026}, parity-obliviousness breaks this equivalence: classical encodings reduce to a single-digit form, while quantum encodings retain an additive structure. At the qubit level, retrieval is incompatible with this structure and yields no quantum advantage, whereas exclusion remains compatible and yields a strictly positive advantage unique to POREC. For general prime $m$, the optimal qubit performance with two-outcome projective measurements gives a quantum-to-noncontextual gap that scales linearly relative to the REC gap and exceeds that of both parity-oblivious retrieval and standard REC/RAC. Numerics further reveal an ordered dimensional hierarchy, and the exact qubit bound provides a sharp semi-device-independent witness certifying $d \ge 3$ with 
considerable noise robustness. For $m=3$, the best see-saw value at $d=6$ matches the Navascu\'es--Pironio--Ac\'{\i}n (NPA) upper bound~\cite{Navascues2008NPA} within numerical precision.

{\it Parity-Oblivious Random Exclusion Code.--} Parity-oblivious prepare-and-measure tasks for binary symbols provide a natural framework to probe preparation contextuality and quantum advantages~\cite{Spekkens2009POM,Chailloux2016PORAC,Ghorai2018POMOptimal,
Singh_2025,Roy2024GPOC}. In the retrieval setting, multi-symbol $m$-PORAC was introduced by Ambainis \emph{et al.}~\cite{Ambainis2019dPORAC} for a single parity class and later extended to multi-mask constraints by Hameedi \emph{et al.}~\cite[Supplementary Sec.~I]{Hameedi2017dPORAC}.
We restrict to prime $m$, for which $\mathbb{Z}_m$ forms a field, making all parity classes uniform. Parity-obliviousness imposes equal-weight constraints, and every nonzero mask is surjective, enabling an exact Fourier-analytic derivation of the classical (preparation-noncontextual) bound~\cite{POREC_SM} (see Supplemental Material Sec.~I). 


Operationally, Alice receives a uniformly random string $x = x_1 x_2 \dots x_n \in \mathbb{Z}_m^n$, and Bob receives a uniformly random index $y \in \{1,\dots,n\}$. Alice encodes $x$ into a message (classical or quantum) sent to Bob, subject to \emph{parity-obliviousness}; she may communicate information about $x$ except any information about multi-digit parities.
Formally, for masks $r \in \mathbb{Z}_m^n$ with Hamming weight $\mathrm{wt}(r)\ge 2$ (the number of nonzero components of $r$), defining
\begin{equation}\label{eq:Parity}
\Pi_r(x) = r \cdot x \pmod m,
\end{equation}
ensures that $\Pi_r$ probes only multi-digit linear correlations; allowing $\mathrm{wt}(r)=1$ would forbid all single-digit information and trivialize the task, whereas $\mathrm{wt}(r)\ge 2$ removes only joint information while preserving access to individual digits. Parity-obliviousness then requires that the average encoding over each class $C_k^{(r)} = \{x : r \cdot x \equiv k \pmod m\}$ is independent of $k$, so that no information about any such parity can be inferred.

We now define the figure of merit. Under parity-obliviousness, the success probability for retrieval ($b = x_y$) is~\cite{Hameedi2017dPORAC,Ambainis2019dPORAC}
\begin{equation}
P_{\mathrm{PORAC}} = \frac{1}{n m^n} \sum_x \sum_{y=1}^n P(b = x_y \mid x, y).
\end{equation}
To contrast this with \emph{exclusion}, we consider the objective $b \neq x_y$, with success probability
\begin{equation}\label{eq:POREC_figure_of_merit}
P_{\mathrm{POREC}} = \frac{1}{n m^n} \sum_x \sum_{y=1}^n P(b \neq x_y \mid x, y).
\end{equation}
We denote this $n$-digit, $m$-symbol setting as the $(n,m)$ parity-oblivious random exclusion task.

In the quantum setting, inputs are encoded into states $\rho_x \in \mathcal{D}(\mathbb{C}^d)$, and we define the parity-class averaged states $\rho^{(r)}_k = \tfrac{1}{|C_k^{(r)}|} \sum_{x \in C_k^{(r)}} \rho_x$. Parity-obliviousness then requires
\begin{equation}\label{eq:rho_parity_classes}
\rho^{(r)}_k = \rho^{(r)}_{k'}, \quad \forall\, k,k' \in \mathbb{Z}_m.
\end{equation}
ensuring parity-obliviousness at the level of preparations. The corresponding quantum success probability is
\begin{equation}\label{eq:POREC_figure_of_merit_quantum}
P_{\mathrm{POREC}}^{\mathrm{Q}}(d)
= \frac{1}{n m^n} \sum_{x\in\mathbb{Z}_m^n} \sum_{y=1}^n P(b \neq x_y \mid x, y),
\end{equation}
while for PORAC,
\begin{equation}
P_{\mathrm{PORAC}}^{\mathrm{Q}}(d)
= \frac{1}{n m^n} \sum_{x\in\mathbb{Z}_m^n} \sum_{y=1}^n P(b = x_y \mid x, y).
\end{equation}
The RAC and REC figures of merit have the same form but without the parity-obliviousness constraint in Eq.~(\ref{eq:Parity}).

{\it Classical and noncontextual bounds.--}
We now establish the optimal benchmark attainable by all classical and preparation-noncontextual theories under full parity-obliviousness.

\begin{theorem}[Classical and noncontextual bounds]
For $n$ digits and prime $m$, any parity-oblivious prepare-and-measure protocol can encode information about at most one digit of $x \in \mathbb{Z}_m^n$. Hence, the optimal preparation-noncontextual success probability of POREC is
\begin{equation}
P_{\mathrm{POREC}}^{\mathrm{NC}}
= 1-\frac{n-1}{mn}.
\label{eq:nc_bounds}
\end{equation}
\label{th:Classical_and_noncontextual_bounds}
\end{theorem}

\textit{Proof sketch.}—
The argument proceeds in three steps.

\emph{(i) Structural collapse:}—
For a classical encoding $p(M|x)$ with $x\in\mathbb{Z}_m^n$, write
$p(M|x)=\sum_{r\in\mathbb{Z}_m^n}\hat p(M,r)\,\omega^{r\cdot x}$ with
$\omega=e^{2\pi i/m}$. For prime $m$, parity-obliviousness enforces
equality of averages over all classes with $\mathrm{wt}(r)\ge2$, eliminating
all multi-digit Fourier components. The encoding therefore contains only
constant and single-digit terms, reducing to a convex mixture of single-digit
channels (see Supplemental Material Sec.~I~\cite{POREC_SM}).

\emph{(ii) Optimal classical strategy:}—
Within this class, sending a single digit is optimal. For fixed $i$, optimality requires the supports of $q_i(M|l)$ to be disjoint across $l$, so that $M$ identifies $x_i$ deterministically, yielding success $1$ when $y=i$ and $(m-1)/m$ otherwise, with $P_i-P_0=1/(nm)>0$.

\emph{(iii) Noncontextual reduction:}—
Let $\lambda\in\Lambda$ be the ontic state, with preparation distributions $\mu(\lambda|P_x)$ and response functions $\xi(b|y,\lambda)$ such that~\cite{Leifer_2014},
\begin{equation}
P(b|x,y)=\int d\lambda\,\mu(\lambda|P_x)\,\xi(b|y,\lambda),
\end{equation}
where $\int d\lambda\,\mu(\lambda|P_x)=1$, $0\le\xi(b|y,\lambda)\le1$, and $\sum_b \xi(b|y,\lambda)=1$.

For masks $r$ with $\mathrm{wt}(r)\ge 2$, define the coarse-grained preparations
\begin{equation}\label{eq:Parity_class}
P_k^{(r)}=\frac{1}{|C_k^{(r)}|}\sum_{x\in C_k^{(r)}} P_x,\quad
C_k^{(r)}=\{x:\, r\cdot x \equiv k \pmod m\}.
\end{equation}
Parity-obliviousness implies the operational equivalence $P(\cdot|P_k^{(r)},M)=P(\cdot|P_{k'}^{(r)},M)$ for all $M,k,k'$, which, by preparation noncontextuality~\cite{Spekkens2005Contextuality}, yields $\mu(\lambda|P_k^{(r)})=\mu(\lambda|P_{k'}^{(r)})$.

Using linearity in Eq.~\eqref{eq:Parity_class} leads to
\begin{equation}
\sum_{x:\,r\cdot x \equiv k}\mu(\lambda|P_x)
=
\sum_{x:\,r\cdot x \equiv k'}\mu(\lambda|P_x)
\quad \forall\,\lambda,k,k'.
\label{eq:ontic_parity_constraint}
\end{equation}
Using Bayes’ rule with a uniform prior over preparations, $p(P_x|\lambda)=\mu(\lambda|P_x)/\sum_{x'}\mu(\lambda|P_{x'})$. Eq.~\eqref{eq:ontic_parity_constraint} then implies
\begin{equation}
\sum_{x:\,r\cdot x \equiv k} p(P_x|\lambda)
=
\sum_{x:\,r\cdot x \equiv k'} p(P_x|\lambda)
\quad \forall\,\lambda,k,k'.
\end{equation}
Hence, even conditioned on $\lambda$, no multi-digit parity information is available. The posterior $p(P_x|\lambda)$ therefore defines a classical parity-oblivious encoding of the preparation label, which by (i) depends on at most one digit.

Substituting into Eq.~(\ref{eq:POREC_figure_of_merit}),
\begin{equation}
P_{\mathrm{POREC}}=\frac{1}{nm^n}\sum_{x,y}\int d\lambda\,\mu(\lambda|P_x)\,\xi(b\neq x_y|y,\lambda),
\end{equation}
the optimal strategy reduces to encoding a single digit (as in step (ii)), yielding success $1$ when $y=i$ and $(m-1)/m$ otherwise. Hence,
$$P_{\mathrm{POREC}}\le 1-\frac{n-1}{mn}=P^{\mathrm{NC}}_{\mathrm{POREC}}.
$$

Thus, preparation noncontextuality enforces a single-digit reduction; any violation implies a breakdown of Eq.~\eqref{eq:ontic_parity_constraint}, and hence preparation contextuality.

{\it Quantum advantage.--}
We now present the central result for POREC, namely that quantum strategies violate the preparation-noncontextual bound under parity-oblivious constraints. We make this explicit in the minimal nontrivial instance $(n,m)=(2,3)$ by identifying the independent parity masks. For inputs $x=(x_1,x_2)\in\mathbb{Z}_3^2$, parity-obliviousness is imposed for masks $r\in\mathbb{Z}_3^2$ with $\mathrm{wt}(r)\ge 2$. Up to nonzero scalar multiples in $\mathbb{Z}_3$, there are two independent masks, $r=(1,1)$ and $r=(1,2)$, with parity classes:
\begin{align*}
r =(1,1)\Rightarrow\;& \{(0,0),(1,2),(2,1)\},\ \{(0,1),(1,0),(2,2)\},\\
        & \{(0,2),(1,1),(2,0)\},\\[4pt]
r =(1,2)\Rightarrow\;& \{(0,0),(1,1),(2,2)\},\ \{(0,2),(1,0),(2,1)\},\\
        & \{(0,1),(1,2),(2,0)\}.
\end{align*}
Parity-obliviousness requires equal averages over each class (see Supplemental Material Sec.~II~\cite{POREC_SM}), imposing linear constraints on the Bloch vectors and yielding the additive structure below.

\begin{lemma}[Additive Bloch structure]\label{Lemma:additive}
For $n=2$ and prime $m$, any qubit encoding $\{\rho_{x_1x_2}\}$ satisfying full multi-mask parity-obliviousness, written as $\rho_{x_1x_2}=\tfrac{1}{2}(I+\vec n_{x_1x_2}\cdot\vec\sigma)$, admits a decomposition $\vec n_{x_1x_2}=\vec a_{x_1}+\vec b_{x_2}$, where $\vec a_{x_1}$ and $\vec b_{x_2}$ depend only on the respective digits.
\end{lemma}

Without loss of generality, fix the gauge $\sum_{x_1}\vec a_{x_1}=\sum_{x_2}\vec b_{x_2}=0$ (Corollary~1, see Supplemental Material Sec.~II~\cite{POREC_SM}). For $(n,m)=(2,3)$, this gives $\vec n_{x_1x_2}=\vec a_{x_1}+\vec b_{x_2}$ with $\vec a_0+\vec a_1+\vec a_2=0$ and $\vec b_0+\vec b_1+\vec b_2=0$, so the nine states arise from two zero-sum triples with no residual joint dependence on $(x_1,x_2)$—the only structure needed for the next theorem.

\begin{theorem}[Exact qubit optimum]\label{th:theorem_exact_qubit_optimum}
For $(n,m)=(2,3)$, the optimal qubit success probability is
\begin{equation}\label{eq:qubit_optimum_m3}
P_{\mathrm{POREC}}^{Q}(2)=\frac{2}{3}+\frac{1}{3\sqrt2},
\end{equation}
exceeding the noncontextual bound $P_{\mathrm{POREC}}^{\mathrm{NC}}=5/6$. The optimum is attained by complementary Pauli measurements, certifying preparation contextuality.
\end{theorem}

\textit{Proof sketch.}—
Using Lemma~\ref{Lemma:additive}, the success probability reduces to 
$P_{\mathrm{POREC}}^{Q}(2)=\frac{2}{3}-\frac{1}{12}F$, where $F$ is linear in the Bloch vectors and POVM elements (see Supplemental Material Sec.~III~\cite{POREC_SM}). Thus, maximizing the success probability is equivalent to minimizing $F$, and by convexity the optimum is attained with extremal POVMs.

The constraint $|\vec a_{x_1}+\vec b_{x_2}|\le 1$, together with the zero-sum relations, ensures that negative overlaps are balanced by positive ones, which sets the tight bound. At the optimum this reduces to $(\lambda t_*)^2+(\mu s_*)^2\le 1$. Minimizing $F$ is then achieved by aligning each Bloch vector anti-parallel to its measurement direction, giving a linear objective under a quadratic constraint. By Cauchy--Schwarz, $|F|\le 2\sqrt{2}$, which is tight and yields Eq.~\eqref{eq:qubit_optimum_m3}.
\\[0.01pt]

\begin{figure}[t]
    \centering
    \includegraphics[width=0.55\columnwidth]{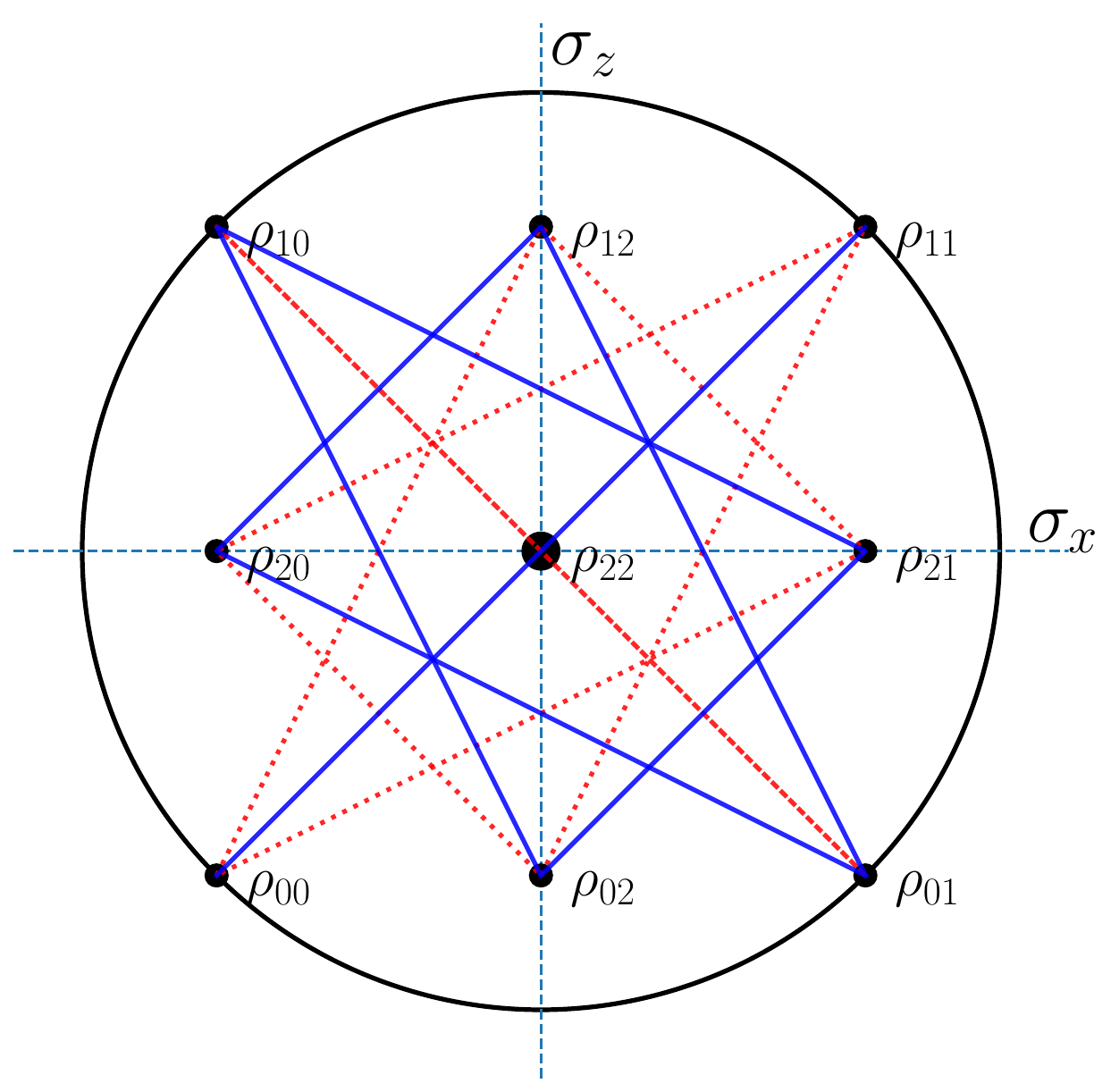}
    \caption{\footnotesize Optimal qubit realization of POREC for $(n,m)=(2,3)$ in the Bloch $x$--$z$ plane. The nine states $\rho_{x_1 x_2}$ form a $3\times3$ grid, with the two digits encoded along orthogonal $\sigma_x$ and $\sigma_z$ directions. Parity classes for masks $(1,1)$ (red dotted) and $(1,2)$ (green solid) each share a common centroid at the maximally mixed state, ensuring full multi-mask parity-obliviousness.}
    \label{fig:bloch_plane_porec}
\end{figure}
An explicit optimal realization is given by the states
$\rho_{x_1 x_2}=\tfrac{1}{2}\!\left(I+\tfrac{f(x_2)}{\sqrt{2}}\sigma_x+\tfrac{f(x_1)}{\sqrt{2}}\sigma_z\right)$
with $f(k)=\delta_{k,1}-\delta_{k,0}$, making the additive structure explicit. The corresponding measurements are two-outcome projective POVMs
$M_{b|y}=\tfrac{1}{2}\bigl(I+(-1)^b\,\hat m_y\cdot\vec{\sigma}\bigr)$ for $b\in\{0,1\}$, with $\hat m_1=\hat z$ and $\hat m_2=\hat x$, and $M_{2|y}=0$. Geometrically, the states lie in the Bloch $x$–$z$ plane, with the two digits encoded along orthogonal directions and measured via $\sigma_z$ and $\sigma_x$. Symmetry fixes a common centroid for all parity classes, ensuring multi-mask parity-obliviousness.

\begin{table}[t]
\caption{\footnotesize
Comparison of optimal classical (RAC, REC) and noncontextual (PORAC, POREC) bounds with quantum values for \((n,m)=(2,3)\).
}
\label{tab:comparison_232}
\centering
\renewcommand{\arraystretch}{1.25}
\begin{ruledtabular}
\begin{tabular}{lccc}
Task & $P_{\mathrm C}, P_{\mathrm{NC}}$ & $P_{\mathrm Q}$ & $\Delta = P_Q - P_{C/NC}$ \\
\hline
\vspace{0.25em}
RAC~\cite{Saha2023}
& $\frac{5}{9}$
& $\frac{4+\sqrt2}{9}$
& $0.0460$
\\
\vspace{0.25em}
REC~\cite{Bae2026}
& $\frac{8}{9}$
& $\frac{7+\sqrt2}{9}$
& $0.0460$
\\
\vspace{1em}
PORAC (single mask)~\cite{Ambainis2019dPORAC}
& $\frac{2}{3}$
& $0.5994$
& $-0.0673$
\\
\vspace{0.25em}
PORAC (multi-mask)~\cite{Hameedi2017dPORAC}
& $\frac{2}{3}$
& $\frac{1}{3}+\frac{1}{3\sqrt2}$
& $-0.0977$
\\
\vspace{0.25em}
\textbf{POREC (this work)}
& $\dfrac{5}{6}$
& $\frac{2}{3}+\frac{1}{3\sqrt2}$
& $\mathbf{0.0690}$
\end{tabular}
\end{ruledtabular}
\renewcommand{\arraystretch}{1}
\end{table}
As summarized in Table~\ref{tab:comparison_232}, parity-obliviousness separates exclusion and retrieval. While PORAC shows no quantum advantage, POREC achieves a strictly positive violation of the noncontextual bound at the qubit level. The corresponding gap is larger by a factor of $3/2$ than in RAC/REC, whereas for PORAC the gap is negative, as no qubit strategy exceeds the noncontextual bound.

\textit{Physical origin of the POREC–PORAC contrast:} Parity-obliviousness reduces all classical and preparation-noncontextual encodings to a single-digit form, while qubit encodings retain the additive structure $\vec n_{x_1x_2}=\vec a_{x_1}+\vec b_{x_2}$. Retrieval requires isolating one component while ignoring the other, which for qubits requires incompatible measurements. This is possible for $m=2$, but for $m\ge3$ the additive structure and finer discrimination make such isolation impossible, so no quantum advantage arises in PORAC. In contrast, exclusion  requires ruling out only the correct value and uses anti-aligned measurements, so a quantum advantage arises in POREC.

This structural difference persists for all prime $m \ge 3$. The additive qubit structure remains compatible with anti-aligned measurements, while retrieval requires identifying a single digit, which becomes increasingly infeasible as $m$ grows. We now show that POREC exhibits a qubit advantage over PORAC for all prime $m \ge 3$, as formalized in the following theorem.
\begin{theorem}[Projective qubit optimum]\label{th:Projective_qubit_optimum}
For $n=2$ and any prime symbol size $m \ge 3$, the optimal success probability over qubit strategies restricted to two-outcome projective measurements is
\begin{equation}
P_{\mathrm{POREC}}^{Q,\mathrm{proj}}(2)
=
\frac{m-1}{m}+\frac{1}{m\sqrt{2}}.
\label{eq:projective_qubit_optimum}
\end{equation}
\end{theorem}

The optimum in Eq.~\eqref{eq:projective_qubit_optimum} is attained by mutually orthogonal measurements with preparations satisfying the parity constraints (see Supplemental Material Sec.~IV~\cite{POREC_SM} for details) and matches the see-saw values in Table~\ref{tab:merged_porec_table_Prob_noise} up to numerical precision ($\sim 10^{-7}$). In particular, for all prime $m \ge 3$, this value exceeds the noncontextual bound, so a qubit strategy witnesses preparation contextuality operationally in POREC.

For $m \ge 5$, whether general $m$-outcome qubit POVMs can exceed this value remains open. Combining Theorem~\ref{th:Projective_qubit_optimum} with the noncontextual bound in Theorem~\ref{th:Classical_and_noncontextual_bounds} [Eq.~\eqref{eq:nc_bounds}] yields the qubit ($d=2$) quantum-to-noncontextual gap $\Delta_{\mathrm{POREC}}$.

For comparison, in REC the classical bound is $1-\frac{1}{m^2}$~\cite{Bae2026}, while the optimal qubit projective value is $1-\frac{2-\sqrt{2}}{m^2}$, giving the gap $\Delta_{\mathrm{REC}}$.

\begin{observation}[Linear enhancement of the gap]
For $n=2$ and prime $m \ge 3$, $\Delta_{\mathrm{POREC}}=\frac{m}{2}\,\Delta_{\mathrm{REC}}$, so the POREC advantage grows linearly with the symbol size.
\end{observation}

The REC expression holds for general $m$, while the POREC bound is derived for prime $m$, allowing a direct comparison. In RAC/REC, both classical and quantum strategies may jointly encode $(x_1,x_2)$~\cite{Saha2023,Bae2026}, subject only to a dimension constraint. By contrast, parity-obliviousness in POREC restricts classical encodings to a single digit and enforces an additive structure for qubits, $\vec n_{x_1x_2}=\vec a_{x_1}+\vec b_{x_2}$. This constraint suppresses classical performance more strongly than the quantum one, thereby amplifying the quantum-to-noncontextual gap.

{\it Finite-dimensional hierarchy.--}
See-saw optimization over $100$ random initializations yields the lower bounds reported in Table~\ref{tab:merged_porec_table_Prob_noise}. For prime $m=3,5,7$, the success probability increases monotonically with the Hilbert-space dimension,
$
P_{\mathrm{POREC}}^{Q}(2) < P_{\mathrm{POREC}}^{Q}(3) < P_{\mathrm{POREC}}^{Q}(4)$, revealing a clear dimension hierarchy already at low $d$. By contrast, PORAC shows its first nontrivial quantum violation only at $d=4$, and even then with a comparatively smaller quantum-to-noncontextual gap~\cite[Supplementary Sec.~I]{Hameedi2017dPORAC}.

For $m=3$, the NPA hierarchy~\cite{Navascues2008NPA} gives level-$4$ and level-$5$ upper bounds $P_{\rm NPA}^{(4)} = 0.919119704805156$ and $P_{\rm NPA}^{(5)} = 0.9191192988879445$, while the best see-saw value at $d=6$ is $0.9191192311615828$. The agreement within $\sim 10^{-7}$ indicates finite-dimensional saturation. As $m$ increases, the size of the NPA moment matrix grows rapidly with the hierarchy level; already for $m=3$, reaching this precision requires level $5$, and higher $m$ becomes substantially more demanding.

{\it Noise robustness and experimental implementation.--}
To investigate experimental feasibility of POREC, we consider white noise via the depolarizing model $\rho_x(\omega)=(1-\omega)\rho_x+\omega\,I/d$, with $\omega\in[0,1]$. A quantum advantage persists when the noisy success probability exceeds the noncontextual bound (Theorem~\ref{th:Classical_and_noncontextual_bounds}), defining the threshold $\omega_c^{(d)}=
\frac{P_d^{Q}-P^{NC}}{P_d^{Q}-(1-1/m)},$
where $P_d^{Q}$ is the quantum success probability at dimension $d$ and $P^{NC}\equiv P^{\mathrm{NC}}_{\mathrm{POREC}}$.

Table~\ref{tab:merged_porec_table_Prob_noise} shows that $\omega_c^{(d)}$ is nonzero for all tested parameters and increases with dimension for fixed $m$, meaning higher-dimensional systems tolerate more noise before the quantum advantage disappears. This is experimentally relevant, as a nonzero $\omega_c^{(d)}$ ensures that preparation contextuality remains certifiable under realistic imperfections. Moreover, parity-obliviousness can be operationally verified by testing the indistinguishability of parity-class ensembles defined in Eq.~(\ref{eq:rho_parity_classes}), via the bound $P_{\mathrm{guess}}^{(r)}=\tfrac{1}{2}+\tfrac{1}{4}\mathrm{Tr}\,|\rho^{(r)}_k-\rho^{(r)}_{k'}|\approx\tfrac{1}{2}$ for all $r,\;k\neq k'$. Since both the success probability and this indistinguishability witness are directly estimable from photon-counting statistics in existing prepare-and-measure platforms~\cite{Spekkens2009POM,Hendrych2012DWExp,Sun2016, Mazurek2016Noncontextuality,Zhang_2022,Sun2026}, POREC is within near-term experimental implementation.

\begin{table}[t]
\centering
\caption{\footnotesize Quantum values \((P_{\mathrm{POREC}}^{Q}(d)\equiv P_d^Q)\) and corresponding critical white-noise thresholds \(\omega_c^{(d)}\) for \((2,m)\) POREC.}
\label{tab:merged_porec_table_Prob_noise}
\renewcommand{\arraystretch}{1.2}
\begin{tabular}{c|c|c|c|c|c|c|c}
\hline
$m$ & Classical 
& \multicolumn{2}{c|}{$d=2$} 
& \multicolumn{2}{c|}{$d=3$} 
& \multicolumn{2}{c}{$d=4$} \\
\cline{3-8}
& 
& $P_2^Q$ & $\omega_c^{(2)}$ 
& $P_3^Q$ & $\omega_c^{(3)}$ 
& $P_4^Q$ & $\omega_c^{(4)}$ \\
\hline
3 & 0.833333 
& 0.902369 & 0.2929 
& 0.911306 & 0.3188 
& 0.916664 & 0.3333 \\
5 & 0.900000 
&  0.941421 & 0.2929 
& 0.946786 & 0.3188 
& 0.958579 & 0.3694 \\
7 & 0.928571 
& 0.958158 & 0.2929
& 0.961989 & 0.3187 
& 0.970413 & 0.3693 \\
\hline
\end{tabular}
\renewcommand{\arraystretch}{1}
\end{table}



{\it Witnessing the Hilbert-space dimension via POREC.--}
Certification of Hilbert-space dimension has been studied from various perspectives. These include device-independent approaches based on observed statistics~\cite{Brunner2008Hilbert, Gallego2010DimensionWitness}. Robustness against detection inefficiencies~\cite{DallArno2012,Bowles2014NonlinearDW} and connections to state discrimination have been explored~\cite{Brunner2013DW}.  
Certain other approaches are based on dynamical and graph-theoretic methods~\cite{Wolf2009Dynamics,Wehner2008DW,Ray2020GraphDW}.
 General lower bounds have been derived in the PM scenario, as well~\cite{Sikora2016}. Moreover, polytope-based approaches rooted in generalized noncontextuality provide facet inequalities that can serve as dimension witnesses~\cite{Hazra2026efficient}. There exist several 
experimental demonstrations of dimension witnessing~\cite{Ahrens2012DimensionExperiment,Hendrych2012DWExp,Sun2016,
Sun2026,Bialecki2024DW}. 

In the context of our present work, the exact qubit optimum (Theorem~\ref{th:theorem_exact_qubit_optimum}) shows that POREC defines a semi-device-independent dimension witness~\cite{Brunner2013DW,Bowles2014NonlinearDW,
deVicente2019DimensionBound}, requiring only observed statistics together with the operational parity-obliviousness constraint. Any violation of the qubit bound certifies higher dimension,
$P_{\mathrm{obs}}>P_{\mathrm{Opt}}\Rightarrow \dim(\mathcal H)\ge 3$ and, for the $(2,3)$ instance,
\begin{equation}
P_{\mathrm{POREC}} > 0.902369
\;\Rightarrow\;
d \ge 3.
\end{equation}
POREC thus provides a clear operational route to dimension certification by combining a communication task with preparation contextuality while retaining an analytic classical bound. It also exhibits strong noise robustness, with critical visibilities $\omega_c^{(2)}=0.293$ and $\omega_c^{(3)}=0.319$, improving on the polytope-based results ($\omega_c^{(2)}=0.272$ and $\omega_c^{(3)}=0.296$)~\cite{Hazra2026efficient}. The robustness of POREC is comparable to that of standard linear dimension witnesses, such as the qubit-versus-classical-bit discrimination task~\cite{DallArno2012}, where a critical detection efficiency $\eta_{qc}=1/\sqrt{2}$ corresponds, under white-noise mixing, to $\omega_c \approx 0.293$.

{\it Discussion.—} In this work we have demonstrated how replacing the task of retrieval
 with exclusion characteristically modifies or reshapes the operational manifestation of preparation contextuality under parity-oblivious constraints. Our results show that parity-obliviousness collapses all classical and preparation-noncontextual encodings to a single-digit form, while admissible quantum encodings retain an additive structure for qubits. The operational consequences are task dependent. In retrieval (PORAC)~\cite{Hameedi2017dPORAC}, measurement incompatibility together with the additive coupling prevents accessing one component independently, eliminating any qubit advantage. In contrast, exclusion (POREC) is compatible with this structure, since ruling out outcomes suffices, yielding a qubit advantage already at the two-digit three-symbol $(2,3)$ level. The equivalence between retrieval and exclusion observed in RAC/REC~\cite{Bae2026} is broken for all prime symbol size $m$ in POREC.  Stronger suppression of classical encodings compared to quantum ones leads to an enhanced quantum-to-noncontextual gap in POREC, providing linear advantage with $m$ relative to the REC gap.

Nontrivial preparation constraints give rise to a dimension-dependent hierarchy,  $P_{\mathrm{POREC}}^{Q}(d) < P_{\mathrm{POREC}}^{Q}(d+1)$, in contrast to the existing results based on preparation contextuality~\cite{Spekkens2009POM,Ambainis2019dPORAC,Hameedi2017dPORAC,
Ghorai2018POMOptimal,Roy2024GPOC,Singh_2025,patra2025contextuality}.  The exact qubit value yields a sharp semi-device-independent witness certifying $d \ge 3$ with analytic noise robustness (for $m=3$, the see-saw value at $d=6$ coincides with the NPA upper bound~\cite{Navascues2008NPA} within numerical precision). In essence, exclusion remains nontrivial even when retrieval becomes effectively classical, identifying POREC as a distinct regime with an enhanced quantum-to-noncontextual gap, favorable for
a variety of applications. The protocol is compatible with current prepare-and-measure platforms~\cite{Ahrens2012DimensionExperiment,Hendrych2012DWExp,Sun2016,
Mazurek2016Noncontextuality,Sun2026} and is robust against noise, supporting near-term implementations of preparation contextuality for dimension certification and cryptographic applications such as high-rate multi-outcome randomness generation~\cite{Random_number_generator_review2017,Brown2021, Hazra2026efficient}, and  quantum key distribution~\cite{Ding2017HDQKD,Tan2021SecureKeyRates,Brown2021, Chaturvedi2021characterising}.   

Our results inspire further studies in several directions. 
Since the noncontextual bound holds for all $(n,m)$, the behaviour seen for $(2,m)$ is not incidental, but the simplest instance of a general phenomenon, motivating a systematic study of $(n>2,m\ge 3)$ POREC scenarios where the advantage persists with new scaling laws. On the other hand,
identifying constraints that further tighten the obtained bounds or modify the scaling remains open. A natural extension is to sequential scenarios, where multiple observers act on the same system~\cite{Silva2015,Shiladitya2016sharing,Debarshi2019facets,Sasmal2018Steering,Brown2020}. As in sequential RACs~\cite{Mohan_2019,Miklin2020unsharp,Anwer2021noiserobust,Abhyoudai2023,
Roy2026}, controlled disturbance can enhance information extraction and enable distributed contextuality. Since sequential measurements can enhance device-independent randomness generation~\cite{Bowles2020boundingsetsof}, sequential POREC offers a route to certifying high-rate multi-outcome randomness under parity-oblivious constraints.

{\it Acknowledgments.--} R.A. acknowledges financial support from the National Science and Technology Council, Taiwan (Grants No. 112-2628-M-006-007-MY4, 114-2811-M-006-069-MY2).

\newpage

\onecolumngrid
\renewcommand{\theequation}{S\arabic{equation}}
\setcounter{equation}{0}
\noindent \begin{center}{\Large \bf Supplemental Material}\end{center}
~\vspace{-0.5cm}

\section{Sec I: Proof of (Theorem 1) optimal classical and preparation-noncontextual bound for prime $m$}

In this section, we derive the optimal classical (preparation-noncontextual) in Theorem 1, the success probability for the $(n,m)$ POREC task with $m$ prime:
$$
P_{\mathrm{POREC}}^{\mathrm{cl}} = 1 - \frac{n-1}{mn}.
$$

\textit{Problem definitions and parity-obliviousness:}
Alice receives an input string $x=(x_1,\dots,x_n)\in\mathbb{Z}_m^n$ uniformly at random and encodes it into a classical message $M$ via a channel $p(M\mid x)$. Bob receives an index $y\in\{1,\dots,n\}$ uniformly at random and must output a value $b\neq x_y$.

Parity-obliviousness requires that for every mask $r\in\mathbb{Z}_m^n$ with $\mathrm{wt}(r)\ge 2$ (at least two  $r_i \ne 0$),
\begin{align}
P(r\cdot X = k \mid M)=\tfrac{1}{m}, \quad \forall\,k\in\mathbb{Z}_m.
\label{eq:parity_oblivious}
\end{align}
Define $C_k^{(r)}=\{x:\,r\cdot x\equiv k\!\!\!\pmod{m}\}$. For prime $m$, $|C_k^{(r)}|=m^{n-1}$ for all nonzero $r$ and all $k$. Applying Bayes’ rule (with a uniform prior over $x$) yields
\begin{align}
\sum_{x\in C_k^{(r)}} p(M\mid x)
=
\sum_{x\in C_{k'}^{(r)}} p(M\mid x),
\quad \forall\,k,k'.
\label{eq:equal_weight_classes}
\end{align}

Let $\omega=e^{2\pi i/m}$ and expand
\begin{align}
p(M\mid x)
=
\sum_{r\in\mathbb{Z}_m^n}\hat p(M,r)\,\omega^{r\cdot x}, 
\quad
\hat p(M,r)=\frac{1}{m^n}\sum_x p(M\mid x)\,\omega^{-r\cdot x}.
\end{align}
For any $r$ with $\mathrm{wt}(r)\ge 2$, group the defining sum by parity classes:
\begin{align}
\hat p(M,r)
=
\frac{1}{m^n}\sum_{k\in\mathbb{Z}_m}
\left(\sum_{x\in C_k^{(r)}} p(M\mid x)\right)\omega^{-k}.
\end{align}
By Eq.~\eqref{eq:equal_weight_classes}, the inner sum is independent of $k$, so it factors out:
\begin{align}
\hat p(M,r)
=
C(M,r)\sum_{k=0}^{m-1}\omega^{-k}
=
0.
\end{align}
Thus all multi-digit Fourier modes vanish, leaving only constant and single-digit terms. The encoding must therefore have the form
$$
p(M\mid x)
= f_0(M) + \sum_{i=1}^n \sum_{l=0}^{m-1} t_i^{(l)}(M)\,\delta_{x_i,l}.
$$
The coefficients satisfy $\sum_l t_i^{(l)}(M)=0$ for each $i$ and $M$, which follows from orthogonality of Fourier characters. To obtain a nonnegative decomposition, define
$$
\alpha_i(M) := \min_l t_i^{(l)}(M) \le 0,
$$
and set
$$
a_{i,l}(M) := t_i^{(l)}(M) - \alpha_i(M) \ge 0, \qquad
a_0(M) := f_0(M) + \sum_i \alpha_i(M).
$$
Then
\begin{equation}\label{eq:nonnegative_decomposition}
p(M\mid x) = a_0(M) + \sum_{i,l} a_{i,l}(M)\,\delta_{x_i,l}.
\end{equation}

For each $i$, there exists $l_i^*$ such that $a_{i,l_i^*}(M)=0$. Evaluating at $x=(l_1^*,\dots,l_n^*)$ gives
$$
p(M\mid x)=a_0(M)\ge 0,
$$
so $a_0(M)\ge 0$.
\smallskip
Define
$$
K_{i,l} := \sum_M a_{i,l}(M), \qquad p_0 := \sum_M a_0(M)\ge 0.
$$
Summing over $M$ yields
$$
p_0 + \sum_{i=1}^n K_{i,x_i} = 1 \quad \forall x.
$$
Comparing inputs differing only at coordinate $i$ shows that $K_{i,l}$ is independent of $l$, so we write $K_{i,l}=K_i$. Hence
$$
p_0 + \sum_i K_i = 1.
$$
Defining
$$
q_0(M) := \frac{a_0(M)}{p_0}, \qquad
q_i(M\mid l) := \frac{a_{i,l}(M)}{K_i},
$$
and using Eq.~\eqref{eq:nonnegative_decomposition}, we obtain
$$
p(M\mid x) = p_0\,q_0(M) + \sum_{i=1}^n K_i\,q_i(M\mid x_i).
$$

By linearity,
$$
P_{\mathrm{succ}} = p_0 P_0 + \sum_{i=1}^n K_i P_i,
$$
where $P_0 = \frac{m-1}{m}$. Since $p_0 + \sum_i K_i = 1$, this is a convex combination, and therefore
$$
P_{\mathrm{succ}} \le \max\{P_0, \max_i P_i\}.
$$

For a fixed $i$,
$$
P_i = \frac{1}{n} P(\mathrm{succ}\mid y=i) + \frac{n-1}{n}\cdot\frac{m-1}{m}.
$$
The second term is fixed. The first term is maximised when $M$ determines $x_i$ uniquely, which requires disjoint supports of $q_i(\cdot\mid l)$ for different $l$. Any overlap reduces the success probability. Thus the optimal choice is deterministic, $M=x_i$, giving
$$
P_i = \frac{1}{n} + \frac{n-1}{n}\cdot\frac{m-1}{m}
= \frac{m-1}{m} + \frac{1}{nm}.
$$

Since $P_i > P_0$, the optimal value satisfies
$$
P_{\mathrm{succ}} \le 1 - \frac{n-1}{mn}.
$$
\emph{Achievability:}
Alice sends $M=x_1$. If $y=1$, Bob outputs any $b\neq M$; if $y\neq 1$, he outputs a fixed $b_0\in\mathbb{Z}_m$. This achieves
\begin{equation}\label{eq:P_classical}
    P_{\mathrm{POREC}}^{\mathrm{cl}}
= \frac{1}{n}\cdot 1 + \frac{n-1}{n}\cdot\frac{m-1}{m}
= 1 - \frac{n-1}{mn}.
\end{equation}

As discussed in the main text, this classical bound coincides with the preparation-noncontextual bound $P_{\mathrm{POREC}}^{\mathrm{NC}}$.
\section{Sec II: Proof of (Lemma 1) Additive structure under parity-obliviousness and zero sum gauge}
\label{app:additive}

In this section, we prove that full parity-obliviousness enforces an additive structure on the preparation Bloch vectors. We consider $(2, m)$  scenario $(x_1,x_2)\in\mathbb{Z}_m^2$, where $m$ is prime, and a family of qubit states $\{\rho_{x_1 x_2}\}$. Writing
$$
\rho_{x_1 x_2} = \tfrac{1}{2}\bigl(I + \vec n_{x_1 x_2}\cdot \vec\sigma\bigr),
$$
we analyze the constraints imposed on the Bloch vectors $\vec n_{x_1 x_2}\in\mathbb{R}^3$ by parity-obliviousness.

\medskip
\subsection{Proof of Lemma 1 (Additive Bloch structure)}
Under full parity-obliviousness, the Bloch vectors admit a decomposition
$$
\vec n_{x_1 x_2} = \vec a_{x_1} + \vec b_{x_2}
\qquad \forall\, x_1,x_2 \in \mathbb{Z}_m.
$$

\begin{proof}
Full parity-obliviousness is defined with respect to all masks $r=(r_1,r_2)\in\mathbb{Z}_m^2$ with $r_1,r_2\neq 0$, through the parity function $\Pi_r(x_1,x_2)=r_1 x_1 + r_2 x_2 \pmod m$.
\smallskip
The condition requires that, for every such mask $r$, the averaged states over each parity class
$$
\rho^{(r)}_k
:= \frac{1}{m}\sum_{x_1,x_2:\, r_1 x_1 + r_2 x_2 = k} \rho_{x_1 x_2}
$$
are independent of $k$.

\smallskip
Since $m$ is prime, every nonzero element of $\mathbb{Z}_m$ is invertible; in particular, $r_2 \neq 0$ implies $r_2^{-1}$ exists, and multiplying the parity equation by $r_2^{-1}$ yields the equivalent form $r_2^{-1} r_1 x_1 + x_2 = r_2^{-1} k$. Defining $\alpha := r_1 r_2^{-1} \in \mathbb{Z}_m^\times,$ this becomes $
\alpha x_1 + x_2 = k',$ where $k' = r_2^{-1}k$ is simply a relabelling of the parity class. Since parity-obliviousness requires independence for all $k$, this relabelling does not affect the condition.

\smallskip
Thus, without loss of generality, it suffices to impose parity-obliviousness for all masks of the form $(\alpha,1)$ with $\alpha \in \mathbb{Z}_m^\times$. For each $\alpha$, the parity classes are therefore described by $\alpha x_1 + x_2 = k \quad (\mathrm{mod}\ m),$ which can be written as $x_2 = k - \alpha x_1.$

Parity-obliviousness then requires that
$$
\rho^{(\alpha)}_k
:= \frac{1}{m}\sum_{x_1 \in \mathbb{Z}_m} \rho_{x_1,\,k-\alpha x_1}
$$
is independent of $k$.

Passing to Bloch vectors using linearity of the map $\rho \mapsto \vec n$, this is equivalent to requiring that
$$
f(k) := \sum_{x_1 \in \mathbb{Z}_m} \vec n_{x_1,\,k-\alpha x_1}
$$
is independent of $k$ for every $\alpha \in \mathbb{Z}_m^\times$.

A function on $\mathbb{Z}_m$ is constant if and only if all its nonzero discrete Fourier coefficients vanish. Applying this component-wise to $f(k)$, we obtain that for every $p \neq 0$,
$$
\sum_{k \in \mathbb{Z}_m} f(k)\,\omega^{pk} = 0,
\qquad \omega := e^{2\pi i/m}.
$$
Substituting the definition of $f(k)$ gives
$$
\sum_{k} \sum_{x_1} \vec n_{x_1,\,k-\alpha x_1}\,\omega^{pk} = 0.
$$

Introduce the change of variables $x_2 = k - \alpha x_1$, equivalently $k = x_2 + \alpha x_1$, which defines a bijection on $\mathbb{Z}_m^2$. The sum becomes
$$
\sum_{x_1,x_2} \vec n_{x_1 x_2}\,\omega^{p(x_2 + \alpha x_1)}
= \sum_{x_1,x_2} \vec n_{x_1 x_2}\,\omega^{(p\alpha)x_1 + p x_2}.
$$

Define the discrete Fourier transform
$$
\widehat{\vec n}_{s,t}
:= \sum_{x_1,x_2 \in \mathbb{Z}_m} \vec n_{x_1 x_2}\,\omega^{s x_1 + t x_2}.
$$
Then parity-obliviousness implies
\begin{equation}\label{eq:fourier_constraint}
\widehat{\vec n}_{p\alpha,\,p} = 0
\qquad \forall\, \alpha \in \mathbb{Z}_m^\times,\; p \neq 0.
\end{equation}

Let $s,t \neq 0$. Since $m$ is prime, $t$ is invertible, so set $p = t$ and $\alpha = s t^{-1}$, giving $p\alpha = s$, and hence by Eq.~\eqref{eq:fourier_constraint}, $\widehat{\vec n}_{s,t} = 0$.

\smallskip
Thus, all mixed Fourier modes vanish. The inverse transform,
$$
\vec n_{x_1 x_2}
= \frac{1}{m^2}\sum_{s,t} \widehat{\vec n}_{s,t}\,\omega^{-(s x_1 + t x_2)}.
$$
Only modes with $s=0$ or $t=0$ surviving,
$$
\vec n_{x_1 x_2}
= \frac{1}{m^2}\sum_{s} \widehat{\vec n}_{s,0}\,\omega^{-s x_1}
+ \frac{1}{m^2}\sum_{t} \widehat{\vec n}_{0,t}\,\omega^{-t x_2}
- \frac{\widehat{\vec n}_{0,0}}{m^2}.
$$

Define
$$
\vec a_{x_1}
:= \frac{1}{m^2}\sum_{s} \widehat{\vec n}_{s,0}\,\omega^{-s x_1}
- \frac{\widehat{\vec n}_{0,0}}{m^2},
\qquad
\vec b_{x_2}
:= \frac{1}{m^2}\sum_{t} \widehat{\vec n}_{0,t}\,\omega^{-t x_2}.
$$
Then

\begin{equation}\label{eq:additive_structure}
    \vec n_{x_1 x_2} = \vec a_{x_1} + \vec b_{x_2},
\end{equation}
which completes the proof.
\end{proof}

Together, these constraints eliminate all joint dependence on $(x_1,x_2)$ and enforce the additive structure. This is precisely the decomposition in Eq.~\eqref{eq:additive_structure}.

\smallskip
\paragraph{\textbf{Remark} (Role of multiple masks and primality):}
A single mask only constrains averages along one family of parallel hyperplanes, so it removes only part of the mixed Fourier structure; correlations that depend jointly on $(x_1,x_2)$ can still remain. To eliminate all such mixed terms, one needs the full set of masks. In particular, for any mixed mode $(s,t)$, the mask with slope $\alpha = s t^{-1}$ removes it. This construction relies on the fact that every nonzero $t$ has an inverse, which holds precisely when $m$ is prime. If $m$ is composite, some elements are not invertible, so certain mixed modes cannot be accessed by any mask, and the additive decomposition may fail.

\medskip
\begin{corollary}[Zero-sum gauge]
\label{cor:zero_sum}
The decomposition can always be chosen such that
\begin{equation}\label{eq:zero_sum}
\sum_{x_1} \vec a_{x_1} = 0,
\qquad
\sum_{x_2} \vec b_{x_2} = 0.
\end{equation}
\end{corollary}

\begin{proof}
The decomposition is not unique, since it is invariant under the transformation
$$
\vec a_{x_1} \to \vec a_{x_1} - \vec c,
\qquad
\vec b_{x_2} \to \vec b_{x_2} + \vec c,
$$
for any fixed vector $\vec c$. Choosing
$\vec c = \frac{1}{m}\sum_{x_1} \vec a_{x_1}$ ensures that $\sum_{x_1} \vec a_{x_1} = 0$.

\smallskip
Summing the decomposition over all inputs gives
$$
\sum_{x_1,x_2} \vec n_{x_1 x_2}
= m \sum_{x_1} \vec a_{x_1} + m \sum_{x_2} \vec b_{x_2}.
$$
Hence,
$$
\sum_{x_2} \vec b_{x_2}
= \frac{1}{m}\sum_{x_1,x_2} \vec n_{x_1 x_2}.
$$

For our figure of merit, only differences between states are relevant (we will see details in the next section). Therefore, we can shift all preparations by a constant Bloch vector without affecting either parity-obliviousness or the success probability. In particular, we may choose the centroid of the ensemble to be the maximally mixed state:
$$
\frac{1}{m^2}\sum_{x_1,x_2} \rho_{x_1 x_2} = \frac{I}{2},
$$
which is equivalent to
$$
\sum_{x_1,x_2} \vec n_{x_1 x_2} = 0.
$$
Under this choice, it follows that
$$
\sum_{x_2} \vec b_{x_2} = 0,
$$
as required by Eq.~\eqref{eq:zero_sum}, completing the proof.
\end{proof}

\textbf{$\boldsymbol{(2,3)}$ case.}
For $m=3$, full parity-obliviousness involves two independent masks, $(1,1)$ and $(1,2)$, each defining a partition of $\mathbb{Z}_3^2$ into three parity classes.

\emph{Mask $(1,1)$:} The parity $k = x_1 + x_2 \ (\mathrm{mod}\ 3)$ yields

\begin{align*}
 k=0:\{(0,0),(1,2),(2,1)\},\; k=1:\{(0,1),(1,0),(2,2)\},\; k=2:\{(0,2),(1,1),(2,0)\}.   
\end{align*}

\smallskip
which imply
$$
\vec n_{00}+\vec n_{12}+\vec n_{21}
=
\vec n_{01}+\vec n_{10}+\vec n_{22}
=
\vec n_{02}+\vec n_{11}+\vec n_{20}.
$$

\emph{Mask $(1,2)$:}
The parity $k = x_1 + 2x_2 \ (\mathrm{mod}\ 3)$ yields
\begin{align*}
k=0 : \{(0,0),(1,1),(2,2)\},\;
k=1 : \{(0,2),(1,0),(2,1)\},\;
k=2 : \{(0,1),(1,2),(2,0)\},
\end{align*}
which imply
$$
\vec n_{00}+\vec n_{11}+\vec n_{22}
=
\vec n_{02}+\vec n_{10}+\vec n_{21}
=
\vec n_{01}+\vec n_{12}+\vec n_{20}.
$$

\medskip
Taken together, these two independent sets of constraints eliminate all joint dependence on $(x_1,x_2)$ and enforce the additive structure
$$
\vec n_{x_1 x_2} = \vec a_{x_1} + \vec b_{x_2}.
$$

Imposing the zero-sum gauge,
$$
\vec a_0 + \vec a_1 + \vec a_2 = 0,
\qquad
\vec b_0 + \vec b_1 + \vec b_2 = 0,
$$
each family $\{\vec a_{x_1}\}$ and $\{\vec b_{x_2}\}$ forms a closed triangle in Bloch space. The nine Bloch vectors are then obtained as pairwise sums
$$
\vec n_{x_1 x_2} = \vec a_{x_1} + \vec b_{x_2},
$$
so the full encoding is generated by translating one triangle by the vertices of the other. This geometric picture makes explicit that all correlations arise from independent single-digit contributions, with no residual joint structure.
\section{Sec III: Proof of (Theorem 2) qubit upper bound for the $(2,3)$ POREC/PORAC task and saturating strategy}\label{app:qubit_bound}

In this section, we derive the tight upper bound for qubit ($d=2$) strategies in the $(2,3)$ parity-oblivious exclusion task. The same derivation applies, with a simple change of alignment, to the corresponding parity-oblivious retrieval (PORAC) task. As described in the main text, Alice receives an input string $x=(x_1,x_2)\in \mathbb{Z}_3^2$, while Bob receives an index $y\in\{1,2\}$. In the PORAC task, Bob is required to output the value of the selected digit, i.e., $b = x_y$. In the POREC task, his goal is instead to output a symbol $b$ that differs from $x_y$.

\smallskip
The corresponding average success probabilities for qubit ($d=2$) strategies are
\begin{equation}\label{eq:PORAC_success}
    P_{\mathrm{PORAC}}^{\mathrm{Q}}(d=2)
=
\frac{1}{2\cdot 3^2}
\sum_{x\in\mathbb{Z}_3^2}
\sum_{y=1}^2
P(b=x_y\mid x,y),
\end{equation}

and
$$
P_{\mathrm{POREC}}^{\mathrm{Q}}(d=2)
=
\frac{1}{2\cdot 3^2}
\sum_{x\in\mathbb{Z}_3^2}
\sum_{y=1}^2
P(b\neq x_y\mid x,y).
$$

Using the Born rule $P(b=x_y\mid x,y)=\operatorname{Tr}(\rho_x M_{x_y|y}),$
and noting that exclusion succeeds exactly when $b\neq x_y$, we rewrite
\begin{equation}\label{eq:POREC_reduced}
P_{\mathrm{POREC}}^{\mathrm{Q}}(d=2)
=
1-\frac{1}{18}
\sum_{x_1,x_2=0}^{2}
\Bigl[
\operatorname{Tr}(\rho_{x_1x_2}M_{x_1|1})
+
\operatorname{Tr}(\rho_{x_1x_2}M_{x_2|2})
\Bigr].
\end{equation}

Hence, introduce the marginal operators,
$$
R_{x_1}:=\sum_{x_2}\rho_{x_1x_2},
\qquad
C_{x_2}:=\sum_{x_1}\rho_{x_1x_2}.
$$
Rewriting Eq.~\eqref{eq:POREC_reduced} in terms of these marginals, we obtain
\begin{equation}\label{eq:POREC_marginal}
P_{\mathrm{POREC}}^{\mathrm{Q}}(d=2)
=
1-\frac{1}{18}
\Bigl[
\sum_{x_1}\operatorname{Tr}(R_{x_1}M_{x_1|1})
+
\sum_{x_2}\operatorname{Tr}(C_{x_2}M_{x_2|2})
\Bigr].
\end{equation}

Using the completeness relation $M_{2|y}=I-M_{0|y}-M_{1|y}$ together with
$\operatorname{Tr}(R_2)=\operatorname{Tr}(C_2)=3$, Eq.~\eqref{eq:POREC_marginal} can be rewritten as
\begin{equation}\label{eq:POREC_Sxy}
P_{\mathrm{POREC}}^{\mathrm{Q}}(d=2)
=
\frac{2}{3}
-\frac{1}{18}
\sum_{\substack{x=0,1\\ y=1,2}}
\operatorname{Tr}\bigl[M_{x|y}\,S_{x|y}\bigr],
\end{equation}
where $S_{x|1}=R_x-R_2$ and $S_{x|2}=C_x-C_2$.

\smallskip
We now use the additive Bloch decomposition derived in Sec.~II,
$$
\rho_{x_1x_2}
=
\frac{1}{2}\bigl(I+(\vec a_{x_1}+\vec b_{x_2})\cdot\vec\sigma\bigr),
\qquad
\sum_{x_1}\vec a_{x_1}=0,\quad
\sum_{x_2}\vec b_{x_2}=0,
$$
which gives
$$
R_{x_1}-R_2=\frac{3}{2}(\vec a_{x_1}-\vec a_2)\cdot\vec\sigma,
\qquad
C_{x_2}-C_2=\frac{3}{2}(\vec b_{x_2}-\vec b_2)\cdot\vec\sigma.
$$

We parametrize the measurement operators as
\begin{equation}\label{eq:Measurement_constraint}
M_{k|y}
=
\frac{t_{k|y}}{2}(I+\hat m_{k|y}\cdot\vec\sigma),
\qquad
\sum_k t_{k|y}=2,\quad
\sum_k t_{k|y}\hat m_{k|y}=0,
\end{equation}
and use the identity $\operatorname{Tr}[M(\vec v\cdot\vec\sigma)]=t\,\hat m\cdot\vec v.$
Substituting into Eq.~\eqref{eq:POREC_Sxy}, we obtain
\begin{equation}\label{eq:POREC_final_F}
P_{\mathrm{POREC}}^{\mathrm{Q}}(d=2)
=
\frac{2}{3}-\frac{1}{12}F,
\end{equation}
where
\begin{equation}\label{eq:F}
F=
\sum_{x_1 =0}^{2} t_{x_1 |1}(\vec a_{x_1}\cdot\hat m_{x_1 |1})
+
\sum_{x_2 =0}^{2} t_{x_2 |2}(\vec b_{x_2}\cdot\hat m_{x_2 |2}).    
\end{equation}

From Eq.~\eqref{eq:POREC_final_F}, maximizing $P_{\mathrm{POREC}}^{\mathrm{Q}}$ is equivalent to minimizing $F$.
\smallskip

\begin{proposition}\label{Proposition1}
At the optimum, up to relabeling of outcomes, the preparation Bloch vectors can be chosen as
\begin{equation}
\vec{a}_{x_1}=-\lambda\,t_{x_1|1}\,\hat{m}_{x_1|1}, 
\qquad
\vec{b}_{x_2}=-\mu\,t_{x_2|2}\,\hat{m}_{x_2|2},
\end{equation}
for some constants $\lambda,\mu>0$.
\end{proposition}

\begin{proof}
The objective function $F$ is affine in both the preparations and the POVM elements. Hence the optimum is attained at extremal points of the convex set of admissible strategies, so it suffices to consider extremal POVMs.

For fixed measurement directions, each contribution $t_{x|y}(\vec v\cdot \hat m_{x|y})$ is minimized when $\vec v$ is anti-parallel to $\hat m_{x|y}$. Therefore, without loss of optimality, the preparation vectors can be written as
\begin{equation}
\vec a_{x_1}=-\lambda_{x_1}\hat m_{x_1|1},\qquad
\vec b_{x_2}=-\mu_{x_2}\hat m_{x_2|2},
\end{equation}
with $\lambda_{x_1},\mu_{x_2}\ge 0$.

The constraints $\sum_{x_1}\vec a_{x_1}=0$ and $\sum_{x_2}\vec b_{x_2}=0$ imply
\begin{equation}
\sum_{x_1}\lambda_{x_1}\hat m_{x_1|1}=0,\qquad
\sum_{x_2}\mu_{x_2}\hat m_{x_2|2}=0.
\end{equation}

 For an extremal three-outcome qubit POVM, the effects are rank-1 and satisfy $\sum_k t_{k|y}\hat m_{k|y}=0$. This implies that the Bloch vectors $\{\hat m_{k|y}\}$ lie in a plane and are not collinear; hence the space of linear dependencies among them is one-dimensional.

Comparing with the POVM completeness relation
\begin{equation}
\sum_{x_1} t_{x_1|1}\hat m_{x_1|1}=0,
\end{equation}
it follows that the coefficients must be proportional, i.e.,
\begin{equation}
\lambda_{x_1}=\lambda\,t_{x_1|1}.
\end{equation}

An identical argument yields
\begin{equation}
\mu_{x_2}=\mu\,t_{x_2|2}.
\end{equation}
Since $\lambda=\mu=0$ gives $F=0$, which is strictly suboptimal, we conclude $\lambda,\mu>0$.

The preparation vectors must satisfy $|\vec a_{x_1}+\vec b_{x_2}|\le 1$ for all $(x_1,x_2)$. This couples the two sets of vectors and constrains the admissible measurement directions. The resulting geometric structure will be exploited in the proof of Theorem~2.
\end{proof}

\medskip
\noindent
\section{Proof of Theorem 2}
We prove that $|F|\le 2\sqrt{2}$ and that this bound is tight.
\begin{proof}
Recall
\begin{align*}
F = -(\lambda A + \mu B), \qquad
A = \sum_{x_1} t_{x_1|1}^2, \quad
B = \sum_{x_2} t_{x_2|2}^2,    
\end{align*}
so that $|F| = \lambda A + \mu B$.

\smallskip
Let $t_* = \max_{x_1} t_{x_1|1}$ and $s_* = \max_{x_2} t_{x_2|2}$. Since $\sum_{x_1} t_{x_1|1} = 2$ and $t_{x_1|1} \le t_*$,
\begin{align*}
A = \sum_{x_1} t_{x_1|1}^2 \le t_* \sum_{x_1} t_{x_1|1} = 2t_*,    
\end{align*}
with equality iff $t_{x_1|1} \in \{0, t_*\}$. Similarly $B \le 2s_*$. Hence
\begin{equation}\label{eq:F_bound}
|F| \le 2(\lambda t_* + \mu s_*).    
\end{equation}
Since $|F|$ is convex in the weights over a compact simplex, its maximum is attained at an extreme point, i.e., $t_{x|y}\in\{0,t_*\}$. It suffices to consider the saturating case.

\smallskip
Squaring the Bloch-sphere constraint $\lvert \vec a_{x_1}+\vec b_{x_2}\rvert \le 1$ and using Proposition~\ref{Proposition1}, we obtain Eq.~(\ref{eq:norm_constraint}).

\begin{equation}\label{eq:norm_constraint}
    (\lambda t_{x_1|1})^2 + (\mu t_{x_2|2})^2
+ 2\lambda t_{x_1|1}\mu t_{x_2|2}\, c_{x_1 x_2} \le 1,
\end{equation}

where $c_{x_1 x_2} = \hat m_{x_1|1} \cdot \hat m_{x_2|2}$.

\smallskip
Fix $x_2^*$ with $t_{x_2^*|2} = s_*$. From the measurement constraint Eq.~(\ref{eq:Measurement_constraint}),
\begin{align*}
\sum_{x_1} t_{x_1|1}\hat m_{x_1|1} = 0,
\end{align*}
so
\begin{align*}
\sum_{x_1} t_{x_1|1}\, c_{x_1 x_2^*} = 0.
\end{align*}
Under saturation, $t_{x_1|1}=t_*$ on the support
$S := \{x_1 : t_{x_1|1}>0\}$, hence
\begin{align*}
\sum_{x_1 \in S} c_{x_1 x_2^*} = 0.
\end{align*}
Therefore the coefficients cannot all be negative, and there exists
$x_1' \in S$ such that $c_{x_1' x_2^*} \ge 0$.

Evaluating Eq.~(\ref{eq:norm_constraint}) at $(x_1',x_2^*)$, where
$t_{x_1'|1}=t_*$ and $t_{x_2^*|2}=s_*$, gives
\begin{align*}
(\lambda t_*)^2 + (\mu s_*)^2
+ 2\lambda t_* \mu s_*\, c_{x_1' x_2^*} \le 1,
\end{align*}
which implies
\begin{align*}
(\lambda t_*)^2 + (\mu s_*)^2 \le 1.
\end{align*}

\begin{align*}
\sum_{x_1 \in S} c_{x_1 x_2^*}=0,\; \exists\, c_{x_1 x_2^*}<0
\;\Rightarrow\; \exists\, c_{x_1' x_2^*}>0
\;\Rightarrow\;
(\lambda t_*)^2+(\mu s_*)^2 < 1.
\end{align*}
Hence equality can hold only when $c_{x_1 x_2^*}=0$ on the support.

\smallskip

By Cauchy-Schwarz to Eq.~(\ref{eq:F_bound}),
\begin{align*}
\lambda t_* + \mu s_* \le \sqrt{2}\sqrt{(\lambda t_*)^2 + (\mu s_*)^2} \le \sqrt{2},    
\end{align*}

and therefore
\begin{align*}
|F| \le 2\sqrt{2}.    
\end{align*}

\end{proof}
\paragraph{Saturation of the bound.}
The bound $|F|\le 2\sqrt{2}$ is tight. Saturation of $A\le 2t_*$ and $B\le 2s_*$ requires
$t_{x_1|1}\in\{0,t_*\}$ and $t_{x_2|2}\in\{0,s_*\}$, so the weights are uniform on their supports; for qubits this corresponds to $t_{0|y}=t_{1|y}=1$ and $t_{2|y}=0$, i.e., only two outcomes contribute and the third is null. Together with Eq.~(\ref{eq:Measurement_constraint}), this enforces $(\lambda t_*)^2+(\mu s_*)^2=1$ and vanishing overlaps on the support, $c_{x_1 x_2}=0$. Hence the bound is achieved by two incompatible measurements with orthogonal Bloch directions, e.g., mutually orthogonal Pauli measurements. Saturation can also be realized with asymmetric choices (such as a suitably aligned trine and a projective POVM), but we focus on this symmetric projective construction.

\smallskip
From the zero-sum constraints, the third Bloch vectors are fixed as
\begin{align*}
\vec a_2 &= -\vec a_0 - \vec a_1 = 0, \\
\vec b_2 &= -\vec b_0 - \vec b_1 = 0.
\end{align*}
Consider the additive decomposition $\vec n_{x_1 x_2} = \vec a_{x_1} + \vec b_{x_2}$ and choose
\begin{align*}
\vec a_0 &= -\frac{\hat z}{\sqrt{2}}, \quad
\vec a_1 = \frac{\hat z}{\sqrt{2}}, \\
\vec b_0 &= -\frac{\hat x}{\sqrt{2}}, \quad
\vec b_1 = \frac{\hat x}{\sqrt{2}}.
\end{align*}

The corresponding three-outcome POVMs are
\begin{align*}
M_{0|1} &= \tfrac{1}{2}(\mathbb{I} + \vec{\sigma}\cdot \hat z), \quad
M_{1|1} = \tfrac{1}{2}(\mathbb{I} - \vec{\sigma}\cdot \hat z), \quad
M_{2|1} = 0, \\
M_{0|2} &= \tfrac{1}{2}(\mathbb{I} + \vec{\sigma}\cdot \hat x), \quad
M_{1|2} = \tfrac{1}{2}(\mathbb{I} - \vec{\sigma}\cdot \hat x), \quad
M_{2|2} = 0,
\end{align*}
which satisfy $\sum_b M_{b|y} = \mathbb{I}$. The Bloch directions are
\begin{align*}
\hat m_{0|1} &= \hat z,\quad \hat m_{1|1}=-\hat z, \\
\hat m_{0|2} &= \hat x,\quad \hat m_{1|2}=-\hat x.
\end{align*}

Substituting into
\begin{align*}
F &= \sum_{x_1} t_{x_1|1}\,\vec a_{x_1}\cdot \hat m_{x_1|1}
+ \sum_{x_2} t_{x_2|2}\,\vec b_{x_2}\cdot \hat m_{x_2|2},
\end{align*}
with $t_{0|y}=t_{1|y}=1$, we obtain
\begin{align*}
F &= -\frac{1}{\sqrt{2}} - \frac{1}{\sqrt{2}}
     -\frac{1}{\sqrt{2}} - \frac{1}{\sqrt{2}}
   = -2\sqrt{2}.
\end{align*}

Substituting into
\begin{align*}
P_{\mathrm{POREC}}^{\mathrm{Q}}(d=2)
&=\tfrac{2}{3}-\tfrac{1}{12}F
\end{align*}
gives
\begin{align*}
P_{\mathrm{POREC}}^{\mathrm{Q}}(d=2)
&= \frac{2}{3}+\frac{1}{3\sqrt{2}}.
\end{align*}

For the PORAC task, using Eq.~(\ref{eq:PORAC_success}), the same construction applies with parallel (rather than anti-parallel) alignment, which flips the sign of $F$. The bound is thus attained at $F=+2\sqrt{2}$, yielding
\begin{align*}
P_{\mathrm{PORAC}}^{\mathrm{Q}}(d=2)
&= \frac{1}{3}+\frac{1}{3\sqrt{2}}.
\end{align*}

\section{Sec IV: Proof of (Theorem 3), $(2,m)$ POREC qubit bound under two-outcome projective measurements}\label{app:Theorem3_proof}
As discussed in the main text, Alice receives an input
$x=(x_1,x_2)\in \mathbb{Z}_m^2$, and Bob receives a choice
$y\in\{1,2\}$. In the PORAC task, Bob must output $b=x_y$, whereas in the POREC task he must output any symbol different from $x_y$.

The corresponding average success probability for qubit strategies is
\begin{equation}\label{eq:POREC_def}
P_{\mathrm{POREC}}^{\mathrm{Q}}(d=2)
=
\frac{1}{2m^2}
\sum_{x\in\mathbb{Z}_m^2}
\sum_{y=1}^2
P(b\neq x_y\mid x,y).
\end{equation}

Using the Born rule $P(b=x_y\mid x,y)=\operatorname{Tr}(\rho_x M_{x_y|y})$ and noting that exclusion succeeds iff $b\neq x_y$, Eq.~\eqref{eq:POREC_def} becomes
\begin{align*}
P_{\mathrm{POREC}}^{\mathrm{Q}}(d=2)
&=
1-\frac{1}{2m^2}
\sum_{x_1,x_2}
\Big[
\operatorname{Tr}(\rho_{x_1x_2}M_{x_1|1})
+
\operatorname{Tr}(\rho_{x_1x_2}M_{x_2|2})
\Big].
\end{align*}

Define the marginals
\begin{equation}\label{eq:marginals}
R_{x_1}:=\sum_{x_2}\rho_{x_1x_2}, \qquad
C_{x_2}:=\sum_{x_1}\rho_{x_1x_2}.
\end{equation}
Reordering the sums yields
\begin{align*}
P_{\mathrm{POREC}}^{\mathrm{Q}}(d=2)
&=
1-\frac{1}{2m^2}
\left[
\sum_{x_1}\operatorname{Tr}(R_{x_1}M_{x_1|1})
+
\sum_{x_2}\operatorname{Tr}(C_{x_2}M_{x_2|2})
\right].
\end{align*}

Using POVM completeness,
\begin{equation}\label{eq:povm_complete}
M_{m-1|y}=I-\sum_{k=0}^{m-2}M_{k|y},
\end{equation}
we obtain
\begin{align*}
P_{\mathrm{POREC}}^{\mathrm{Q}}(d=2)
&=
\frac{m-1}{m}
-
\frac{1}{2m^2}
\Bigg[
\sum_{x_1=0}^{m-2}
\operatorname{Tr}\!\left(M_{x_1|1}(R_{x_1}-R_{m-1})\right) \\
&\qquad\qquad
+
\sum_{x_2=0}^{m-2}
\operatorname{Tr}\!\left(M_{x_2|2}(C_{x_2}-C_{m-1})\right)
\Bigg].
\end{align*}

For $n=2$ and prime $m$, parity-obliviousness enforces an additive Bloch decomposition (see Sec.~II), giving
\begin{equation}\label{eq:bloch_diff}
R_{x_1}-R_{m-1}
=
\frac{m}{2}(\vec a_{x_1}-\vec a_{m-1})\cdot\vec\sigma,
\quad
C_{x_2}-C_{m-1}
=
\frac{m}{2}(\vec b_{x_2}-\vec b_{m-1})\cdot\vec\sigma.
\end{equation}

Substituting Eq.~\eqref{eq:bloch_diff},
\begin{align*}
P_{\mathrm{POREC}}^{\mathrm{Q}}(d=2)
&=
\frac{m-1}{m}
-
\frac{1}{4m}
\Bigg[
\sum_{x_1=0}^{m-2}
\operatorname{Tr}\!\left(M_{x_1|1}(\vec a_{x_1}-\vec a_{m-1})\cdot\vec\sigma\right) \\
&\qquad\qquad
+
\sum_{x_2=0}^{m-2}
\operatorname{Tr}\!\left(M_{x_2|2}(\vec b_{x_2}-\vec b_{m-1})\cdot\vec\sigma\right)
\Bigg].
\end{align*}

Specializing to two-outcome projective measurements,
\begin{equation}\label{eq:projective}
M_{0|y}=\tfrac{1}{2}(I+\hat m_y\cdot\vec\sigma), \quad
M_{1|y}=\tfrac{1}{2}(I-\hat m_y\cdot\vec\sigma), \quad
M_{i|y}=0 \;\; (i\ge 2),
\end{equation}
the expression simplifies to
\begin{equation}\label{eq:reduced_form}
P_{\mathrm{POREC}}^{Q}(d=2)
=
\frac{m-1}{m}
-\frac{1}{4m}
\Big[
\hat m_1\cdot(\vec a_0-\vec a_1)
+
\hat m_2\cdot(\vec b_0-\vec b_1)
\Big].
\end{equation}

From Eq.~\eqref{eq:reduced_form}, maximizing $P_{\mathrm{POREC}}^{Q}$ is equivalent to minimizing
\begin{equation}\label{eq:F_def}
F :=
\hat m_1\cdot(\vec a_0-\vec a_1)
+
\hat m_2\cdot(\vec b_0-\vec b_1).
\end{equation}
From Eq.~\eqref{eq:F_def}, define
\begin{align*}
\tilde{A} := |\vec a_0 - \vec a_1|, \qquad
\tilde{B} := |\vec b_0 - \vec b_1|.
\end{align*}
Using $|\hat m_1|=|\hat m_2|=1$, Eq.~\eqref{eq:F_def} implies $F \ge -\bigl(\tilde{A} + \tilde{B}\bigr),$ with equality when $\hat m_1$ and $\hat m_2$ are anti-parallel to the respective differences.

To bound $\tilde{A}^2 + \tilde{B}^2$, we use the Bloch-sphere constraint
$|\vec a_{x_1}+\vec b_{x_2}|\le 1$, which holds for all $(x_1,x_2)$ and in particular for $(x_1,x_2)\in\{0,1\}^2$:
\begin{align*}
&|\vec a_0+\vec b_0|^2
+|\vec a_0+\vec b_1|^2
+|\vec a_1+\vec b_0|^2
+|\vec a_1+\vec b_1|^2 \le 4.
\end{align*}
Applying the parallelogram identity twice gives
\begin{align*}
&4\left|
\frac{\vec a_0+\vec a_1+\vec b_0+\vec b_1}{2}
\right|^2
+
|\vec a_0-\vec a_1|^2
+
|\vec b_0-\vec b_1|^2 \le 4.
\end{align*}
Dropping the non-negative first term yields
\begin{equation}\label{eq:Atilde_bound}
\tilde{A}^2 + \tilde{B}^2 \le 4.
\end{equation}

By Cauchy-Schwarz and Eq.~\eqref{eq:Atilde_bound},
\begin{align*}
\tilde{A} + \tilde{B}
\le \sqrt{2(\tilde{A}^2 + \tilde{B}^2)}
\le 2\sqrt{2},
\end{align*}
and hence $F \ge -2\sqrt{2}.
$

Substituting this into Eq.~\eqref{eq:reduced_form} gives
\begin{align*}
P_{\mathrm{POREC}}^{Q, \mathrm{Proj}}(d=2)
\le
\frac{m-1}{m}
+
\frac{1}{m\sqrt{2}}.
\end{align*}

The bound is saturated by binary projective measurements along orthogonal Bloch directions, e.g.\ $\hat m_1=\hat z$ and $\hat m_2=\hat x$, with $\vec a_0=-\hat z/\sqrt{2}$, $\vec a_1=\hat z/\sqrt{2}$, $\vec b_0=-\hat x/\sqrt{2}$, and $\vec b_1=\hat x/\sqrt{2}$.
\end{document}